\documentclass[acmsmall,nonacm,screen]{acmart}
\renewcommand{\footnotetextauthorsaddresses}[1]{}
\AtBeginDocument{\hypersetup{allcolors=blue}}

\usepackage[T1]{fontenc}
\usepackage{mathtools}
\usepackage{bm}
\usepackage{tikz}
\usepackage{forest}
\usetikzlibrary{arrows.meta}
\usepackage{booktabs} 
\usepackage{graphicx}
\usepackage[ruled]{algorithm2e} 

\SetAlFnt{\small}
\SetAlCapFnt{\small}
\SetAlCapNameFnt{\small}
\IncMargin{-\parindent}

\usepackage{amsthm}

\newcounter{exctr}
\renewcommand{\theexctr}{\arabic{exctr}}

\newtheorem{theorem}{Theorem}

\newtheorem{lemma}{Lemma}
\newtheorem{proposition}{Proposition}
\newtheorem{corollary}{Corollary}

\makeatletter
\renewcommand{\thmhead}[3]{%
	\thmheadfont{#1\ #2}%
	\thm@notefont{: [#3]}%
}
\makeatother

\newcommand{\E}{\mathbb E}

\newcommand{\Prb}{\mathbb{P}}

\newcommand{\IA}{\mathrm{IA}}
\newcommand{\DA}{\mathrm{DA}}
\newcommand{\rk}{\operatorname{rk}}

\usepackage{microtype}
\setcitestyle{authoryear}

\title{\noindent
Asymptotic Equivalence of Immediate and Deferred Acceptance
}

\author{Josu\'e Ortega}
\affiliation{%
	\institution{Queen's University Belfast}
	\country{UK}	}

\addtocontents{toc}{\protect\setcounter{tocdepth}{-1}}

\begin{abstract}
	\vspace*{2em}
Immediate Acceptance (IA, also known as the Boston mechanism) is commonly used to assign students to schools because, despite its manipulability and instability, it produces a Pareto-efficient matching with truthful reports, unlike student-proposing Deferred Acceptance (DA). 
In this paper, we ask: does IA produce meaningfully better average ranks than DA, conditional on truth-telling? 
We show that, in i.i.d.\ one-to-one random markets, IA's expected average rank is asymptotically $\log n$, just like DA's. 
Therefore, IA’s Pareto efficiency does not translate into a first-order improvement in student placement.
This conclusion extends to variations of IA as well as to many-to-one markets.

\vspace*{1em}
\vspace*{17em}
\end{abstract}
{\protect\setcounter{tocdepth}{3}}

\begin{document}

\begin{titlepage}

\maketitle


\thispagestyle{empty}
\end{titlepage}

\setlength{\parskip}{6pt plus 2pt minus 1pt}

\section{Introduction}
In a school choice problem, an education authority allocates scarce school seats among students, who submit preferences over schools and receive priorities according to predetermined criteria such as proximity, sibling attendance, or academic performance. The central design question is which mechanism should be used to allocate these seats. 

Immediate Acceptance (IA) is a natural candidate to do this. Students apply to their most preferred school, schools decide among competing applications on the basis of priorities, rejecting all who apply once the school is full, and rejected students apply to their most preferred school that has not rejected them. This mechanism is central to the literature because it is widely used and produces a Pareto-efficient outcome under truthful preference revelation. A well-known modification of this algorithm allows schools to consider applications even when they are already full, so that schools can reject previously accepted students in favour of higher-priority late applicants. This modified algorithm is student-proposing Deferred Acceptance (DA), which is strategy-proof and respects priorities, but is Pareto-inefficient with high probability \citep{ortega2026large}. 

Although IA is Pareto-efficient and DA is not, conditional on truth-telling, we do not know how much better IA's placement is in quantitative terms. We tackle this question in this paper. We show that the two mechanisms grow at the same rate: as the market grows, both assign the average student a school of rank asymptotic to $\log n$, so IA's efficiency yields no first-order advantage.

To be precise, we consider a random market with $n$ students and $n$ schools in which students' preferences are drawn uniformly and independently at random, with no assumption made on how priorities are formed. In this random market, classical and influential results from the 1970s by \citet{wilson1972} and \citet{knuth1976} show, through beautiful and elementary arguments, that DA produces an expected average rank for students of $\log n$.\footnote{These results were later strengthened by \cite{pittel1989}, using more advanced techniques.} This means that the average placement grows in large markets (slowly, but grows), which is in sharp contrast with the constant expected average rank that can be achieved by some Pareto-efficient mechanisms \citep{nikzad2022rank, ortega2023cost}. 

Let us explain Wilson's and Knuth's seminal ideas before we proceed. DA, as originally formulated by \citet{gale1962}, has students apply to schools simultaneously, but this is not necessary: one can instead order the students arbitrarily and let each propose to her favourite school that has not yet rejected her, with schools accepting applicants as they arrive and rejected students returning to any position in the queue. \citet{mcvitie1971} show that, for every order of students, this sequential algorithm is equivalent to Gale-Shapley's. Since DA terminates once every school has received at least one application, and since each student applies from her most preferred school to the least, the total number of applications directly determines the average rank. That total number of applications mirrors the classical coupon-collector problem: if each breakfast cereal box contains one of $n$ equally likely coupons, how many cereal boxes do I need to buy to collect all $n$ coupons? Said differently, if each application targets one of $n$ schools randomly, how many applications do I need for each school to receive at least one application?
The answer is $nH_n$, where $H_n \coloneqq \sum_{i=1}^{n} \tfrac{1}{i}$ is the $n$-th harmonic number.\footnote{This is an almost-deterministic result, since the expected collection time is sharply concentrated around its expected value \citep{motwani1995}.} The average student therefore makes about $H_n \sim \log n$ applications and is assigned a school of rank roughly $\log n$.
This argument yields only an upper bound: coupon collection permits a school to appear more than $n$ times. Knuth's contribution is to show that this difference is immaterial at leading order. Deleting repeated draws could in principle let students finish much faster, but most of the coupon collector's time is spent obtaining the last few schools, and by that stage students have explored so little of their lists that repetitions are rare. Therefore, DA's expected average rank is asymptotically 
$\log n$, matching the upper bound.

Nonetheless, the expected average rank remained unknown, because IA has a different and seemingly more complex structure than DA. In particular, in IA applications cannot be made sequentially in arbitrary order because a student may be rejected by a school irrespective of her priority simply because she applied too late. Instead, to analyze IA's average rank, \cite{pritchard2023asymptotic} obtain recursive limiting values of the probability that a student is assigned to her $k$-th-ranked school. Yet,  IA's expected average rank remains unknown. We settle this question with a different and simpler approach, in the spirit of Wilson and Knuth. We show that IA can be equivalently implemented by a sequential algorithm in which students apply to schools one at a time, in any given order, as long as students are not allowed to apply again until every student active in the round has applied once. Then, once the batch of students has applied, schools make acceptance decisions on the basis of their priorities, and unlucky students can apply again.

Using this sequential version of IA restores Wilson's coupon-collector analogy. For each student, generate an infinite sequence of independent uniform draws from the schools and read the sequence until the first school she has not previously tried appears; that school becomes her next application. Pooling all raw draws produces an ordinary coupon collector, and the first global appearance of every school must generate an IA application. The expected number of applications made before every school has received one is therefore at most $nH_n$. Completing the final round adds at most $n$ further applications, yielding the upper bound $nH_n+n$.

The lower bound shows that discarding repeated draws cannot save a first-order fraction of the coupon-collection time. Consider the coupled raw-draw process until at most
\[
a_n=\left\lceil(\log n)^2\right\rceil
\]
schools remain unseen. The number of raw draws required to reach this threshold is $n\log n-o(n\log n)$. Let
\[
s_n=\left\lceil\frac{2n}{\log n}\right\rceil,
\]
and let $A_n$ be the event that IA reaches the threshold by the end of round $s_n$.

On $A_n$, every application made before the threshold comes from a student who has previously tried at most $s_n-1$ schools. The conditional expected number of raw draws needed to produce each such application is therefore at most
\[
\frac{1}{1-s_n/n}=1+o(1).
\]
Writing $p_n=\Prb(A_n)$, this event contributes at least
\[
\bigl(p_n-o(1)\bigr)n\log n
\]
applications in expectation.

On $A_n^c$, more than $a_n$ schools remain unseen, and hence vacant, throughout the first $s_n$ rounds. Because the market is balanced, at least $a_n$ students remain unmatched and apply in each of those rounds. IA therefore makes at least
\[
a_ns_n=(2+o(1))n\log n
\]
applications on this event. Combining the two cases gives
\[
\frac{\E[T_n^{\IA}]}{n\log n}
\ge
p_n+2(1-p_n)-o(1)
\ge
1-o(1).
\]
Together with the upper bound, this proves that IA's expected average rank is asymptotic to $\log n$. The raw-draw coupling follows Knuth's idea, but the quick-versus-slow argument is specific to IA: if coverage is quick, discarded repetitions are negligible; if coverage is slow, IA's round structure itself forces sufficiently many applications.

Therefore, while IA may give students more preferred schools than DA, its advantage becomes small compared with the overall level of average rank as the market grows.

In the Appendix, we show that the same leading-rank result holds for a IA-variant in which students skip already full schools, sometimes called IA with skips or adaptive Boston \citep{miralles2009school,harless2019immediate,mennle2021partial}. We also show that IA and DA share the same leading term in many-to-one markets with bounded quotas. In particular, with $m$ schools of common quota $q$, both mechanisms have expected average rank asymptotic to $(\log m)/q$.

\paragraph{Positioning in the Literature} \citet{abdulkadirouglu2003} first document IA usage in the city of Boston and comment on its shortcomings with regard to manipulation and justified envy empirically and theoretically. Since then, a large literature has studied its properties in theory, in the lab and in the field. 

Every equilibrium of IA's preference revelation game is weakly Pareto dominated by DA \citep{ergin2006}, and therefore IA's average rank can be significantly worse when accounting for strategic behavior \citep{ortegascwe}. Interestingly, in settings with cardinal preferences and incomplete information, \citet{miralles2009school}, \citet{abdulkadirouglu2011resolving}, \citet{featherstone2016boston} and \citet{akyol2024bayesian} show that IA can dominate DA in ex-ante and interim welfare. And there is evidence that switching from IA to DA would reduce welfare in practice \citep{he2015gaming,calsamiglia2020structural,agarwal2018demand}.

The closest paper to ours is \citet{pritchard2023asymptotic}. They study IA under the same i.i.d.\ preference environment, but impose a common priority order across schools. For every fixed rank \(k\), they derive limiting assignment probabilities, including how these probabilities vary with a student's position in the priority order. For IA, the aggregate rank shares are summarized by the recursion
\[
\omega_1=1,
\qquad
\omega_{k+1}=\omega_k e^{-\omega_k},
\]
where $\omega_k$ is the limiting fraction of students who remain unmatched at the beginning of round $k$, and $\omega_k-\omega_{k+1}$ is therefore the limiting fraction assigned their $k$-th choice. In particular, the limiting first-choice share is $1-e^{-1}\approx0.63$.

These fixed-rank limits do not determine expected average rank. They describe the probability of receiving each fixed rank as $n$ grows, but do not control the tail at ranks that themselves grow with $n$. Because $\omega_k\sim 1/k$, this tail is not summable, and a vanishing fraction of students assigned to increasingly poor ranks can determine both the order and the leading coefficient of average rank.
Based on simulations, Pritchard and Wilson conjecture that IA and IA with skips have expected average ranks of order $\log n$, but with leading constants strictly below that of Serial Dictatorship. Theorem~\ref{thm:main} and Theorem~\ref{prop:skip} show instead that all three leading coefficients equal one. Our results also hold uniformly over arbitrary school-priority profiles.

\section{Model and Main Result}
\label{sec:model}

A school-choice problem $P$ consists of a set of students $I$ and a set of schools $S$. Throughout the main analysis, $|I|=|S|=n$ and every school has one seat. Each student $i\in I$ has a strict preference $\succ_i$ over the schools, and each school $s\in S$ has a strict priority $\triangleright_s$ over the students.

A matching $\mu$ is a bijection from $I$ to $S$. We denote by $\mu_i$ the school assigned to student $i$. The rank function
\[
\rk_i(s)=|\{s'\in S:s'\succ_i s\}|+1
\]
assigns rank one to student $i$'s most-preferred school and rank $n$ to her least-preferred school.

The Immediate Acceptance mechanism (IA) proceeds in rounds. In the first round, every student applies to her most-preferred school. Each school permanently accepts its highest-priority applicant and rejects the rest. In every subsequent round, each rejected student applies to the next school in her preference list, and each school that remains vacant permanently accepts its highest-priority applicant in that round. We use $\IA(P)$ to denote the resulting matching and $\IA_i(P)$ to denote the school assigned to student $i$.

The student-proposing Deferred Acceptance mechanism (DA) differs only in that acceptances are tentative. Each school holds its highest-priority applicant and may later displace that student when a higher-priority applicant arrives. We use $\DA(P)$ and $\DA_i(P)$ analogously.

For a mechanism $A\in\{\IA,\DA\}$, define its average student rank in problem $P$ by
\begin{equation}\label{eq:average-rank}
	\overline{\rk}^{A}(P)
	=\frac1n\sum_{i\in I}\rk_i[A_i(P)].
\end{equation}
We study a random market in which student preferences are drawn independently and uniformly from the \(n!\) strict orders of the schools. School priorities are arbitrary and fixed. Let \(\Pi_n\) denote the set of all strict priority profiles, and let \(\E_{\triangleright}\) denote expectation over student preferences for a fixed priority profile \(\triangleright\in\Pi_n\). When no confusion can arise, we suppress the problem \(P\) and write \(\overline{\rk}^{A}_n\). With this notation, we can now state our main result.

\begin{theorem}
	\label{thm:main}
	For every sequence of school-priority profiles \((\triangleright_n)\),
	
	\[
	\lim_{n\to\infty}
	\frac{
		\E_{\triangleright_n}
		\left[\overline{\rk}^{\IA}_n\right]
	}{
		\log n
	}
	=1.
	\]
\end{theorem}

The result does not depend on how priorities vary across markets. Priorities affect which students remain active and therefore the realized application path, but the coupon-collection bounds hold after every possible history.

A direct implication of Theorem~\ref{thm:main}, together with the aforementioned results by \citet{wilson1972, knuth1976} and \citet{pittel1989}, is the asymptotic equivalence of expected average ranks under IA and DA.
\begin{corollary}
	\label{cor:DA}
	For every sequence of school-priority profiles $(\triangleright_n)$,
	
	\[
	\lim_{n\to\infty}
	\frac{
		\E_{\triangleright_n}
		\left[\overline{\rk}^{\IA}_n\right]
	}{
		\E_{\triangleright_n}
		\left[\overline{\rk}^{\DA}_n\right]
	}
	=1.
	\]
\end{corollary}

Corollary~\ref{cor:DA} is a relative first-order statement: the difference between the two expected average ranks is $o(\log n)$, but the result does not assert that their additive difference converges to zero.

We proceed to prove Theorem~\ref{thm:main} in the next Section. 
\section{Proof of Theorem \ref{thm:main}}
\label{sec:sequential}

\subsection{A sequential implementation}
\label{sec:sequential-implementation}

Fix a realized school-choice problem. At the beginning of each IA round, order the unmatched students in any way and reveal their applications one at a time. Schools make no decisions while these applications are being revealed. Once every student who was active at the beginning of the round has applied, each vacant school permanently accepts its highest-priority applicant in that round. The rejected students become active in the next round.

This sequential implementation produces exactly the same matching as the usual simultaneous description of IA. Revealing applications one at a time does not change the set of students applying to each school in a round, and priorities therefore select the same winners. Sequential IA is only an accounting device.

The connection with Wilson's sequential implementation of DA is immediate. Wilson allows proposals to be processed one at a time and lets a rejected student propose again whenever she returns to the queue. IA can also be processed one application at a time, but its rounds must be respected: nobody may apply twice until every student active in the round has applied once. We now show that this restriction does not change the leading coupon-collection cost.

Let $T_n^{\IA}$ denote the total number of genuine IA applications. A student assigned her $r$-th choice makes exactly $r$ applications. Hence
\begin{equation}\label{eq:rank-identity}
	T_n^{\IA}
	=\sum_{i\in I}\rk_i[\IA_i(P)]
	=n\overline{\rk}^{\IA}_n.
\end{equation}
Thus, finding the expected average rank is equivalent to counting applications.

\subsection{The upper bound}
\label{sec:upper}

Call a school \emph{unseen} if it has not yet received any application. Immediately before an application is revealed, let $U$ be the number of unseen schools, let $K$ be the number of schools already tried by the current student, and let $\mathcal F$ contain the complete revealed history.

\begin{lemma}\label{lem:hazard}
	Conditional on $\mathcal F$,
	\begin{equation}\label{eq:hazard}
		\Prb(\text{the next application reaches an unseen school}\mid\mathcal F)
		=\frac{U}{n-K}
		\ge\frac{U}{n}.
	\end{equation}
\end{lemma}

\begin{proof}
	Every unseen school belongs to the current student's unexposed preference-list suffix: if she had previously applied to it, the school would no longer be unseen. Conditional on the exposed prefix of a uniformly random preference order, her next school is uniform among the $n-K$ schools she has not tried. Exactly $U$ of those schools are unseen.
\end{proof}

Let $D_n^{\IA}$ be the number of genuine applications required until every school has received at least one application. When $u$ schools remain unseen, Lemma~\ref{lem:hazard} implies that the next application discovers a new school with probability at least $u/n$. The waiting time to reduce the number of unseen schools from $u$ to $u-1$ is therefore stochastically dominated by a geometric random variable with mean $n/u$. Consequently,
\begin{equation}\label{eq:discovery-upper}
	\E[D_n^{\IA}]
	\le\sum_{u=1}^n\frac{n}{u}
	=nH_n.
\end{equation}

When the final unseen school receives its first application, IA may still need to reveal the remaining applications in that round before schools make their decisions. This adds at most $n$ applications. Therefore
\begin{equation}\label{eq:upper}
	T_n^{\IA}\le D_n^{\IA}+n
	\qquad\text{and hence}\qquad
	\E[T_n^{\IA}]\le nH_n+n.
\end{equation}
Together with \eqref{eq:rank-identity}, this proves the upper bound in Theorem~\ref{thm:main}.

\subsection{The lower bound}
\label{sec:lower}

The previous argument compares IA with a coupon collector that allows repetitions. A genuine IA student never applies twice to the same school, so IA can be faster. The lower bound shows that this advantage changes only lower-order terms.

For every student $i$, generate an infinite sequence
\[
X_{i1},X_{i2},\ldots
\]
of independent uniform draws from the schools. Construct student $i$'s preference order from the order in which distinct schools first appear in this sequence. Whenever sequential IA calls student $i$, read her sequence until the first school that has not appeared previously in her own stream. That school is her next genuine application. We call every raw draw, including a repeated school that is ignored, an \emph{amnesiac tick}.

At each tick, the unread draw in the selected student's stream is independent and uniform over the schools. The pooled sequence of ticks is therefore an ordinary i.i.d.\ coupon collector, even though the identity of the student whose stream is read depends on the IA history. Moreover, the first global appearance of any school is necessarily a genuine application to that school. Thus, the schools unseen by the amnesiac clock are exactly the schools that have never received an IA application.

Let $C_{n,a}$ denote the number of amnesiac ticks required until at most $a$ schools remain unseen. This is the usual partial coupon-collection time, and therefore
\begin{equation}\label{eq:partial-coupon}
	C_{n,a}\ \overset{d}=\
	\sum_{u=a+1}^{n}G_u,
	\qquad
	G_u\sim\operatorname{Geom}\!\left(\frac{u}{n}\right),
\end{equation}
where the summands are independent. It follows that
\begin{align}
	\E[C_{n,a}]
	&=n(H_n-H_a), \label{eq:partial-mean}\\
	\operatorname{Var}(C_{n,a})
	&\le n^2\sum_{u>a}\frac1{u^2}
	=O\!\left(\frac{n^2}{a}\right). \label{eq:partial-var}
\end{align}

Set
\begin{equation}\label{eq:cutoffs}
	a_n=\left\lceil(\log n)^2\right\rceil,
	\qquad
	s_n=\left\lceil\frac{2n}{\log n}\right\rceil.
\end{equation}
Let $A_n$ be the event that IA reaches at most $a_n$ unseen schools by the end of round $s_n$, and let $T_{n,a_n}^{\IA}$ be the number of genuine applications made up to and including the application at which this threshold is first reached.

We first consider $A_n$. Enumerate the genuine applications made up to this threshold by $j=1,\ldots,T_{n,a_n}^{\IA}$. Let $K_j$ be the number of schools previously tried by the student making application $j$, and let $W_j$ be the number of amnesiac ticks required to reveal that application. Every application counted on $A_n$ is made no later than round $s_n$, so $K_j\le s_n-1$.

Let $\mathcal H_n$ be the sigma-field generated by the induced preference orders and the resulting IA history. Conditional on $\mathcal H_n$, the event $A_n$, the application path, and the values $K_j$ are fixed. The first-appearance decomposition of an i.i.d.\ stream implies that the waiting time before the next new school remains geometric after conditioning on the induced order: after $K_j$ distinct schools have appeared, a raw draw is new with probability $(n-K_j)/n$, while the identity of the new school is uniform among the remaining schools and independent of the waiting time. Hence
\[
\E[W_j\mid\mathcal H_n]=\frac{n}{n-K_j}.
\]
Since $C_{n,a_n}=\sum_{j=1}^{T_{n,a_n}^{\IA}}W_j$ on $A_n$, summing these conditional expectations yields
\begin{equation}\label{eq:amnesia-comparison}
	\E[C_{n,a_n}\mathbf 1_{A_n}]
	\le
	\frac{1}{1-s_n/n}
	\E[T_{n,a_n}^{\IA}\mathbf 1_{A_n}].
\end{equation}

Equations~\eqref{eq:partial-mean}-\eqref{eq:partial-var} imply
\begin{equation}\label{eq:partial-L1}
	\E\left|
	\frac{C_{n,a_n}}{n\log n}-1
	\right|\longrightarrow0.
\end{equation}
Indeed,
\[
\E[C_{n,a_n}]
=n\log n-2n\log\log n+O(n),
\]
while its standard deviation is $O(n/\log n)$. Writing $p_n=\Prb(A_n)$, equations~\eqref{eq:amnesia-comparison} and \eqref{eq:partial-L1} give
\begin{equation}\label{eq:amnesia-good-event}
	\E[T_n^{\IA}\mathbf 1_{A_n}]
	\ge
	\E[T_{n,a_n}^{\IA}\mathbf 1_{A_n}]
	\ge
	\bigl(p_n-o(1)\bigr)n\log n.
\end{equation}

We next consider $A_n^c$. More than $a_n$ schools then remain unseen throughout each of the first $s_n$ rounds. Every unseen school is vacant. Because the market is balanced, at least $a_n$ students remain unmatched in each of those rounds, and each of them makes one application. Hence
\begin{equation}\label{eq:amnesia-bad-event}
	T_n^{\IA}
	\ge a_ns_n
	=(2+o(1))n\log n
	\qquad\text{on }A_n^c.
\end{equation}

Combining \eqref{eq:amnesia-good-event} and \eqref{eq:amnesia-bad-event},
\[
\frac{\E[T_n^{\IA}]}{n\log n}
\ge p_n+2(1-p_n)-o(1)
\ge1-o(1).
\]
Together with the upper bound \eqref{eq:upper} and the rank identity \eqref{eq:rank-identity}, this proves Theorem~\ref{thm:main}.

\section{Finite-market simulations}
\label{sec:simulations}

To illustrate the finite-market content of Corollary~\ref{cor:DA}, we simulate IA and DA on the same random markets. For each $n\in\{500,1000,\ldots,10000\}$, we draw 500 markets with independently and uniformly distributed student preferences and independent uniformly random priority orders across schools. The same realized preference and priority profile is used for both mechanisms in each replication.

\begin{figure}[t]
    \centering
    \includegraphics[width=\textwidth]{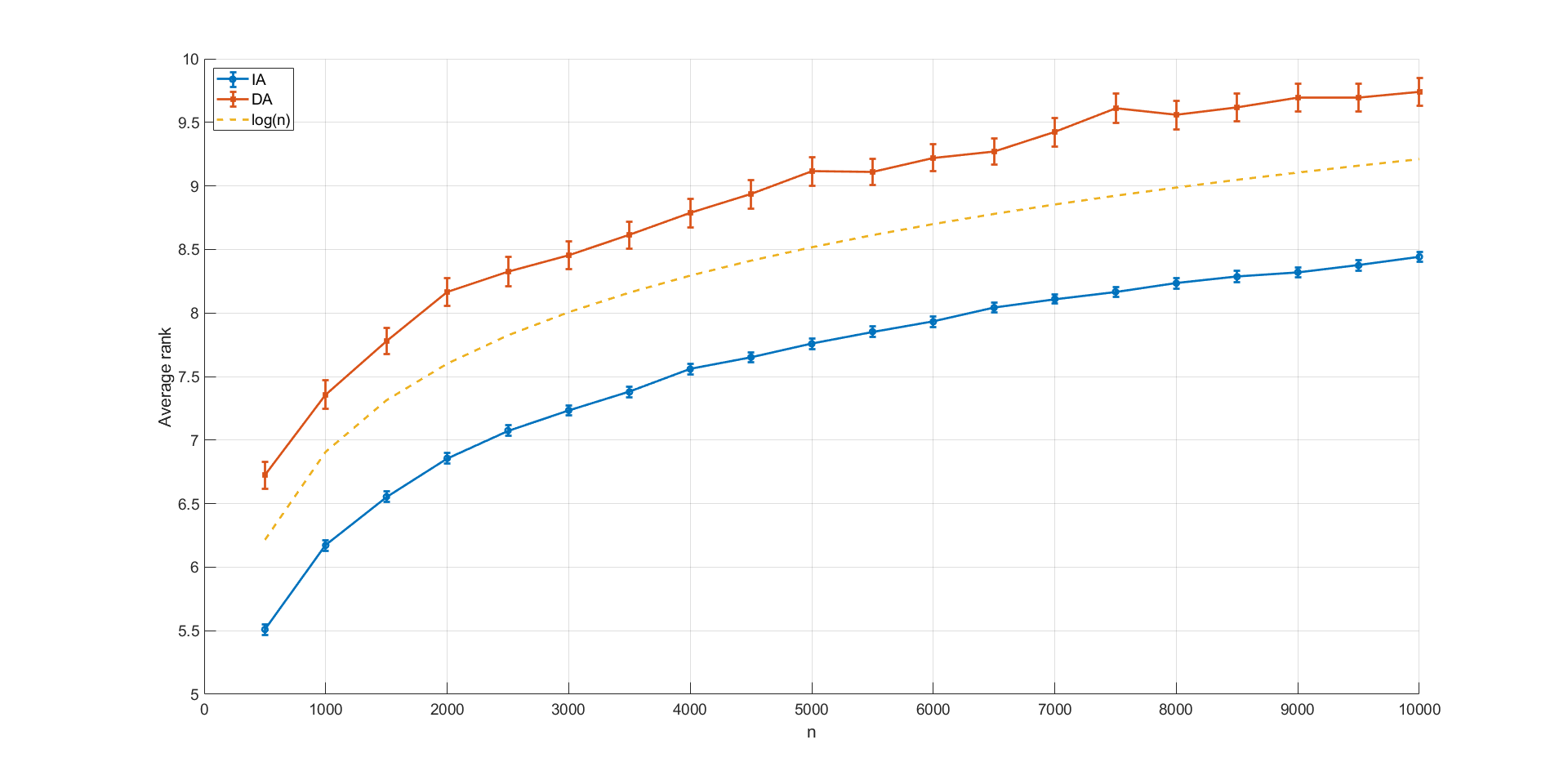}
    \Description{IA and DA's average rank as a function of $n$ (mean over $500$ random markets per $n$).}
    \caption{IA and DA's average rank as a function of $n$ (mean over $500$ random markets per $n$). Bars are 95 percent Monte Carlo confidence intervals.}
    \label{fig:simulations}
\end{figure}

Figure~\ref{fig:simulations} shows a visible finite-market rank advantage for IA. At $n=10{,}000$, the simulated expected average ranks are $8.44$ under IA and $9.74$ under DA, a ratio of $0.867$. Normalizing by $\log n$ gives $0.917$ and $1.057$, respectively. This is the distinction captured by the theorem: the ratio converges to one and the difference is $o(\log n)$, but lower-order differences may remain quantitatively relevant at the market sizes shown.
\section{Extensions}
\label{sec:extensions}

\subsection{IA with skips}
\label{sec:skip}

In the main text, we studied IA when students are allowed to apply to schools that are already full but which will always reject them. A different interpretation of IA in the literature skips those futile applications, so that students can only apply to schools that still have seats by the time they apply to them \citep{miralles2009school,mennle2021partial,harless2019immediate}. We denote this IA modification by $\IA^{\mathrm{skip}}$.
Under this variant, the number of applications no longer equals the rank of the final assignment, because a student may inspect and skip several filled schools before applying. The sequential approach nevertheless continues to work if we count preference-list positions inspected rather than applications.

\begin{theorem}\label{prop:skip}
	For every sequence of school-priority profiles $(\triangleright_n)$,
	\begin{equation}\label{eq:skip-result}
		\lim_{n\to\infty}
		\frac{
			\E_{\triangleright_n}
			\left[\overline{\rk}^{\IA^{\mathrm{skip}}}_n\right]
		}{
			\log n
		}
		=1.
	\end{equation}
\end{theorem}

The intuition is again simple. Suppose that $m$ students and $m$ schools remain at the beginning of a round. Each student's most-preferred remaining school is independently uniform among those $m$ schools. Applications are therefore equivalent to throwing $m$ balls independently into $m$ bins. The students who remain unmatched are exactly the empty bins, so the market contracts by a factor approaching $e^{-1}$ in every round and lasts about $\log n$ rounds. A surviving student scans roughly $n/m$ new positions before finding one of the $m$ vacant schools, so the $m$ active students inspect roughly $n$ positions in total during each substantive round. The total rank sum is therefore again approximately $n\log n$. Appendix~\ref{app:skip} gives the proof.

\subsection{Many-to-one matching}
\label{sec:quotas}

A careful reader will note that DA and IA are mainly used for school choice, where schools have many seats rather than a single one, and may wonder whether Theorem~\ref{thm:main} extends to a many-to-one model in which each school has a quota \(q\) of available seats. The answer is yes. Consider a market with \(m\) schools, each with quota \(q\), and \(mq\) students whose preferences over schools are drawn independently and uniformly at random. In this set-up, we obtain the following result:

\begin{theorem}
	\label{thm:manytoone}
	For every fixed quota \(q\), every sequence of school-priority profiles \((\triangleright_m)\), and each \(A\in\{\IA,\DA\}\),
	\[
	\lim_{m\to\infty}
	\frac{
		\E_{\triangleright_m}
		\left[\overline{\rk}^{A}_m\right]
	}{
		\dfrac{\log m}{q}
	}
	=1.
	\]
	Consequently,
	\[
	\lim_{m\to\infty}
	\frac{
		\E_{\triangleright_m}
		\left[\overline{\rk}^{\IA}_m\right]
	}{
		\E_{\triangleright_m}
		\left[\overline{\rk}^{\DA}_m\right]
	}
	=1.
	\]
\end{theorem}

The proof follows the same coupon-collector logic. Under either mechanism, a student applies down her preference list without ever applying twice to the same school, so the total number of applications equals the sum of assigned-school ranks. The mechanism cannot terminate until every school has received at least \(q\) applications. This is analogous to a coupon collector who must collect \(q\) copies of every coupon.

For fixed \(q\), finding every school for the first time already requires approximately \(m\log m\) applications. Collecting the remaining \(q-1\) applications at each school requires only a lower-order number of additional applications. The total number of applications under either IA or DA is therefore asymptotically \(m\log m\). Since there are \(mq\) students, expected average rank is asymptotically \(\dfrac{\log m}{q}\).

Appendix~\ref{app:quotas} provides the formal proof. It also establishes the more general result in which schools may have different quotas, provided that all quotas remain bounded as the market grows.

\section{Conclusion}
\label{sec:discussion}

Immediate Acceptance does not respect priorities and is manipulable, both in theory and the field \citep{abdulkadirouglu2003,chen2006school}. These flaws can sometimes be arguably tolerated because of its Pareto efficiency. Here, we have shown that IA's efficiency guarantee falls short of the best average rank attainable, which is constant. Not only its average rank grows with the number of students, but it grows at the same rate as the average placement of Deferred Acceptance, and this is without even accounting for how strategic behavior may further distort it. These results do not imply that IA never produces meaningful finite-market gains, but they substantially weaken the case for accepting its incentive and stability costs in exchange for better student placements.

\section*{Acknowledgments}
I am grateful to Herv\'e Moulin for useful conversations on this subject.

\newpage
\bibliographystyle{ACM-Reference-Format}
\bibliography{bibliogr}


\newpage
\appendix
\addtocontents{toc}{\protect\setcounter{tocdepth}{-1}}

\section{Proof for IA with skips}
\label{app:skip}

At the beginning of round $r$, let $I_r$ denote the unmatched students, let $S_r$ denote the vacant schools, and write
\[
M_r=|I_r|=|S_r|.
\]
For every $i\in I_r$, let $K_{i,r}$ be the number of positions already inspected in her preference list. Let $D_{i,r}$ be the number of new positions she inspects before reaching the first school in $S_r$, including that school. Define
\[
S_r^{\mathrm{ins}}=\sum_{i\in I_r}D_{i,r},
\qquad
T_n^{\mathrm{skip}}=\sum_{r\ge1}S_r^{\mathrm{ins}}.
\]
A student's assigned rank equals the total number of positions she inspects. Therefore
\begin{equation}\label{eq:skip-rank}
	T_n^{\mathrm{skip}}
	=n\overline{\rk}^{\IA^{\mathrm{skip}}}_n.
\end{equation}
Let $\mathcal F_r$ denote the history at the beginning of round $r$.

\begin{lemma}\label{lem:skip-exposure}
	Conditional on $\mathcal F_r$ and $M_r=m>0$, the $m$ applications are independent and uniform over the $m$ vacant schools. Moreover,
	\begin{equation}\label{eq:skip-round}
		\E[S_r^{\mathrm{ins}}\mid\mathcal F_r]
		=\frac{m(n+1)-K_r}{m+1},
		\qquad
		K_r=\sum_{i\in I_r}K_{i,r}.
	\end{equation}
\end{lemma}

\begin{proof}
	Every vacant school belongs to every active student's unexposed preference-list suffix. Conditional on the revealed history, these suffixes are independent uniform permutations. Each student's first vacant school is therefore uniform over $S_r$, independently across students. Its position in a suffix of length $n-K_{i,r}$ containing $m$ vacant schools has a negative-hypergeometric mean
	\[
	\frac{n-K_{i,r}+1}{m+1}.
	\]
	Summing over active students gives \eqref{eq:skip-round}.
\end{proof}

Conditional on $M_r=m$, the next survivor count is the number of empty bins after $m$ independent balls are placed uniformly into $m$ bins. In particular,
\begin{equation}\label{eq:skip-contraction}
	\E[M_{r+1}\mid M_r=m]
	=m\left(1-\frac1m\right)^m
	\le\frac{m}{e}.
\end{equation}
Moreover, the number of empty bins is a bounded-differences function of the $m$ ball locations. Since its expectation is asymptotic to $m/e$, for every fixed $\gamma>0$ there is a constant $c_\gamma>0$ such that, for all sufficiently large $m$,
\begin{equation}\label{eq:skip-concentration}
	\Prb\!\left(
	M_{r+1}<e^{-(1+\gamma)}m
	\mid M_r=m
	\right)
	\le e^{-c_\gamma m}.
\end{equation}

For the upper bound, \eqref{eq:skip-round} implies
\[
\E[S_r^{\mathrm{ins}}\mid\mathcal F_r]\le n
\]
whenever the round is non-empty. Iterating \eqref{eq:skip-contraction} gives $\E[M_r]\le ne^{-(r-1)}$. By Markov's inequality, $\Prb(M_r>0)\le\min\{1,\E[M_r]\}$, so the expected number of nonempty rounds is at most
\[
\sum_{r\ge1}\min\{1,ne^{-(r-1)}\}
=\log n+O(1).
\]
Therefore
\begin{equation}\label{eq:skip-upper}
	\E[T_n^{\mathrm{skip}}]
	\le n\log n+O(n).
\end{equation}

For the lower bound, fix $\delta,\eta,\gamma>0$ with $\delta+\eta<1$, and let
\begin{equation}\label{eq:skip-rounds}
	R_n=\left\lfloor
	\frac{1-\delta-\eta}{1+\gamma}\log n
	\right\rfloor.
\end{equation}
Define the stopping round
\[
\tau_n=\inf\{r:M_r<n^\delta\}.
\]
Applying \eqref{eq:skip-concentration} while $M_r\ge n^\delta$ and taking a union bound over the first $R_n$ rounds imply
\begin{equation}\label{eq:skip-survive}
	\Prb(\tau_n>R_n)\longrightarrow1.
\end{equation}

Let
\[
T_r^\tau=\sum_{s\le r}S_s^{\mathrm{ins}}\mathbf 1_{\{s<\tau_n\}}.
\]
On $\{r<\tau_n\}$, we have $M_r\ge n^\delta$. Moreover, the entries previously exposed by the surviving students are among all entries inspected in earlier rounds, so $K_r\le T_{r-1}^\tau$. Equation~\eqref{eq:skip-round} therefore gives
\[
\E\!\left[
S_r^{\mathrm{ins}}\mathbf 1_{\{r<\tau_n\}}
\mid\mathcal F_r
\right]
\ge
\mathbf 1_{\{r<\tau_n\}}
\left[n-O(n^{1-\delta})-n^{-\delta}T_{r-1}^\tau\right].
\]
Summing for $r\le R_n$, taking expectations, and using \eqref{eq:skip-survive},
\[
\E[T_{R_n}^\tau]
\ge
nR_n(1-o(1))
-R_n\,O(n^{1-\delta})
-R_n n^{-\delta}\E[T_{R_n}^\tau].
\]
Since $R_n n^{-\delta}=o(1)$ and $R_n n^{1-\delta}=o(n\log n)$,
\[
\E[T_n^{\mathrm{skip}}]
\ge
\E[T_{R_n}^\tau]
\ge nR_n(1-o(1)).
\]
Thus
\[
\liminf_{n\to\infty}
\frac{\E[T_n^{\mathrm{skip}}]}{n\log n}
\ge
\frac{1-\delta-\eta}{1+\gamma}.
\]
Letting $\delta,\eta,\gamma\to 0$ and combining this inequality with \eqref{eq:skip-upper} and \eqref{eq:skip-rank} proves Theorem~\ref{prop:skip}.
\newpage
\section{Proof for many-to-one markets with bounded quotas}
\label{app:quotas}

There are now $m$ schools. School $s$ has a positive integer quota $q_s$, and the number of students equals the total capacity:
\begin{equation}\label{eq:total-capacity}
	N_m=\sum_{s\in S}q_s.
\end{equation}
Student preferences are independently and uniformly distributed over the $m!$ strict orders of the schools. School priorities remain arbitrary and fixed.

Under IA, each school permanently accepts up to its remaining capacity from the students who apply in a round. Under DA, each school tentatively holds its $q_s$ highest-priority proposers. Under either mechanism, students apply in strict preference order and never apply twice to the same school. A displacement under DA simply causes the displaced student to continue farther down her list. Hence, for $A\in\{\IA,\DA\}$,
\begin{equation}\label{eq:quota-rank}
	T_m^A
	=\sum_{i=1}^{N_m}\rk_i[A_i(P)]
	=N_m\overline{\rk}_m^A,
\end{equation}
where $T_m^A$ denotes the total number of applications.

\begin{proposition}\label{prop:manytoone}
	Fix $Q<\infty$. For every sequence of quota profiles satisfying $1\le q_s\le Q$, every sequence of school-priority profiles $(\triangleright_m)$, and each $A\in\{\IA,\DA\}$,
	\begin{equation}\label{eq:quota-total}
		\lim_{m\to\infty}
		\frac{
			\E_{\triangleright_m}[T_m^A]
		}{
			m\log m
		}
		=1.
	\end{equation}
	More precisely,
	\begin{equation}\label{eq:quota-upper}
		\E_{\triangleright_m}[T_m^A]
		\le
		m\log m+(Q-1)m\log\log m+O_Q(m).
	\end{equation}
	Consequently,
	\begin{equation}\label{eq:quota-average}
		\lim_{m\to\infty}
		\frac{
			\E_{\triangleright_m}[\overline{\rk}_m^A]
		}{
			\dfrac{\log m}{\overline q_m}
		}
		=1,
		\qquad
		\overline q_m=\frac{N_m}{m}.
	\end{equation}
	In particular, if every school has the same fixed quota $q$, then
	\[
	\lim_{m\to\infty}
	\frac{
		\E_{\triangleright_m}[\overline{\rk}_m^{\IA}]
	}{
		\dfrac{\log m}{q}
	}
	=
	\lim_{m\to\infty}
	\frac{
		\E_{\triangleright_m}[\overline{\rk}_m^{\DA}]
	}{
		\dfrac{\log m}{q}
	}
	=1.
	\]
\end{proposition}

The leading term is still the cost of reaching the final schools for the first time. Requiring at most $Q-1$ additional applications at each school affects only the lower-order term. The i.i.d. analogue is the double Dixie cup problem of \citet{newman1960double}; the proof below does not assume that the mechanisms generate i.i.d.\ applications.

\begin{proof}
	We begin with the upper bound. For every school $s$, let $\tau_s$ be the index of its $q_s$-th application, and define
	\[
	C_m=\max_{s\in S}\tau_s.
	\]
	Once each school has received at least its quota of applications, every seat fills at the end of the current round. The only possible slack consists of the remaining applications in that round. Therefore
	\begin{equation}\label{eq:quota-round-slack}
		C_m\le T_m^A\le C_m+N_m.
	\end{equation}
	
	Before school $s$ receives its $q_s$-th application, it is untried by the student making the next application. Under IA, a student who previously applied to $s$ would either have been accepted or would have seen the school fill in that round. Under DA, a school that has received fewer than $q_s$ proposals holds every proposer. In either case, a student who previously applied to $s$ cannot remain active while the school is below quota.
	
	If the current student has already tried $K$ schools, conditional uniformity of her unexposed preference-list suffix gives
	\begin{equation}\label{eq:quota-single-hazard}
		\Prb(\text{the next application is to }s\mid\mathcal F)
		=\frac{1}{m-K}
		\ge\frac1m
	\end{equation}
	whenever fewer than $q_s$ applications have reached $s$.
	
	To make the domination explicit, couple the successive applications with independent uniforms. While school $s$ has received fewer than $q_s$ applications, every comparison success with threshold $1/m$ is also an actual application to $s$, because the latter has conditional probability at least $1/m$. Therefore
	\begin{equation}\label{eq:quota-binomial-domination}
		\Prb(\tau_s>t)
		\le
		\Prb\!\left(
		\operatorname{Bin}\!\left(t,\frac1m\right)<q_s
		\right).
	\end{equation}
	Let $B_{t,m}\sim\operatorname{Bin}(t,1/m)$. Since $q_s\le Q$, a union bound gives
	\begin{equation}\label{eq:quota-union}
		\Prb(C_m>t)
		\le m\Prb(B_{t,m}<Q).
	\end{equation}
	For a constant $c_Q$ depending only on $Q$, and every $t\ge m$,
	\begin{align}
		\Prb(B_{t,m}<Q)
		&=\sum_{j=0}^{Q-1}\binom tjm^{-j}
		\left(1-\frac1m\right)^{t-j} \notag\\
		&\le
		c_Qe^{-t/m}\left(1+\frac tm\right)^{Q-1}.
		\label{eq:quota-binomial-tail}
	\end{align}
	Set
	\begin{equation}\label{eq:quota-threshold}
		t_m=\left\lceil
		m\bigl(\log m+(Q-1)\log\log m\bigr)
		\right\rceil.
	\end{equation}
	For large $m$, the function $e^{-x}(1+x)^{Q-1}$ is decreasing for $x\ge t_m/m$. Equations~\eqref{eq:quota-union}-\eqref{eq:quota-binomial-tail} and an integral comparison imply
	\begin{align*}
		\sum_{t\ge t_m}\Prb(C_m>t)
		&\le
		c_Qm\sum_{t\ge t_m}
		e^{-t/m}\left(1+\frac tm\right)^{Q-1}\\
		&=O_Q(m),
	\end{align*}
	because
	\[
	e^{-t_m/m}\left(1+\frac{t_m}{m}\right)^{Q-1}
	=O_Q\!\left(\frac1m\right).
	\]
	Therefore
	\[
	\E_{\triangleright_m}[C_m]
	=\sum_{t\ge0}\Prb(C_m>t)
	\le t_m+O_Q(m).
	\]
	Since $N_m\le Qm$, equation~\eqref{eq:quota-round-slack} proves \eqref{eq:quota-upper}.
	
	We now prove the matching lower bound. Use the same amnesiac construction with $m$ schools and reveal either mechanism sequentially within its ordinary rounds. For every student, generate an infinite i.i.d.\ uniform sequence of schools and construct her preference order from the first appearances of distinct schools. Raw draws are amnesiac ticks, and repeated schools are ignored until the next untried school appears. The pooled tick sequence is i.i.d.\ uniform, and the first global appearance of every school is its first genuine application under the mechanism.
	
	Let $C_{m,a}$ be the number of ticks required until at most $a$ schools remain unseen. As in \eqref{eq:partial-coupon}-\eqref{eq:partial-var},
	\[
	C_{m,a}\ \overset{d}=\
	\sum_{u=a+1}^{m}G_u,
	\qquad
	G_u\sim\operatorname{Geom}\!\left(\frac{u}{m}\right),
	\]
	with independent summands. For
	\[
	a_m=\left\lceil(\log m)^2\right\rceil,
	\]
	we therefore have
	\begin{equation}\label{eq:quota-partial-L1}
		\E\left|
		\frac{C_{m,a_m}}{m\log m}-1
		\right|\longrightarrow0,
	\end{equation}
	because
	\[
	\E_{\triangleright_m}[C_{m,a_m}]
	=m\log m-2m\log\log m+O(m)
	\]
	and $\operatorname{Var}(C_{m,a_m})=O(m^2/a_m)$.
	
	Set
	\[
	s_m=\left\lceil\frac{2m}{\log m}\right\rceil
	\]
	and let $A_m$ be the event that the mechanism reaches at most $a_m$ unseen schools by the end of round $s_m$. Let $T_{m,a_m}^A$ denote the number of genuine applications made up to and including the application at which this threshold is first reached.
	
	On $A_m$, every such application occurs by round $s_m$. Under both IA and DA, an active student makes at most one application in each round. She has therefore tried at most $s_m-1$ schools. Let $\mathcal H_m$ be the sigma-field generated by the induced preferences and the resulting mechanism history. By the same first-appearance argument used in the proof of Theorem~\ref{thm:main}, conditional on $\mathcal H_m$ the amnesiac waiting time preceding an application has mean $m/(m-K)$, where $K\le s_m-1$ on $A_m$. Summing these conditional means gives
	\begin{equation}\label{eq:quota-amnesia-comparison}
		\E_{\triangleright_m}[C_{m,a_m}\mathbf 1_{A_m}]
		\le
		\frac{1}{1-s_m/m}
		\E_{\triangleright_m}[T_{m,a_m}^A\mathbf 1_{A_m}].
	\end{equation}
	If $p_m=\Prb(A_m)$, equations~\eqref{eq:quota-partial-L1} and \eqref{eq:quota-amnesia-comparison} imply
	\begin{equation}\label{eq:quota-good-event}
		\E_{\triangleright_m}[T_m^A\mathbf 1_{A_m}]
		\ge
		\bigl(p_m-o(1)\bigr)m\log m.
	\end{equation}
	
	On $A_m^c$, more than $a_m$ schools remain unseen throughout each of the first $s_m$ rounds. Every such school has at least one empty seat. Because the market is balanced, the number of active students equals the number of empty seats and is therefore at least the number of unseen schools; every active student makes one application in the round. Therefore
	\begin{equation}\label{eq:quota-bad-event}
		T_m^A
		\ge a_ms_m
		=(2+o(1))m\log m
		\qquad\text{on }A_m^c.
	\end{equation}
	Combining \eqref{eq:quota-good-event} and \eqref{eq:quota-bad-event},
	\[
	\frac{\E_{\triangleright_m}[T_m^A]}{m\log m}
	\ge p_m+2(1-p_m)-o(1)
	\ge1-o(1).
	\]
	Together with \eqref{eq:quota-upper}, this proves \eqref{eq:quota-total}. Equation~\eqref{eq:quota-average} follows from the rank identity \eqref{eq:quota-rank}.
\end{proof}

\end{document}